\documentclass[a4paper,USenglish]{lipics-v2016}
\usepackage{microtype}
\usepackage{xspace}
\usepackage{comment}

\definecolor{blueLink}{rgb}{0,0.2,0.8}
\hypersetup{colorlinks,linkcolor=blueLink,urlcolor=blueLink,citecolor=blueLink}
\newcommand{\lref}[2][]{\hyperref[#2]{#1~\ref*{#2}}}

\newcommand{\ALG}{\textsc{Alg}\xspace}
\newcommand{\OPT}{\textsc{Opt}\xspace}

\newcommand{\dist}{\textsf{dist}}
\newcommand{\atime}{\textsf{atime}}
\newcommand{\weight}{\textsf{weight}}
\newcommand{\budget}{\textsf{budget}}

\newcommand{\wait}{\textsf{wait}}
\newcommand{\cost}{\textsf{cost}}
\newcommand{\costALG}{\textsf{cost}_\mathrm{ALG}}
\newcommand{\costALGNF}{\textsf{cost}_\mathrm{ALG-NF}}
\newcommand{\costOPT}{\textsf{cost}_\mathrm{OPT}}

\newcommand{\C}{\mathcal{C}}
\newcommand{\REAL}{\mathbb{R}}
\newcommand{\X}{\mathcal{X}}
\newcommand{\weightt}{\textsf{ws}}
\newcommand{\leaves}{\textsf{L}}
\newcommand{\size}{\textsf{size}}

\definecolor{islamicgreen}{rgb}{0.0,0.56,0.0}

\theoremstyle{definition}
\newtheorem{observation}[theorem]{Observation}{\bfseries}{\itshape}

\title{A Match in Time Saves Nine: Deterministic Online Matching With Delays\footnote{%
Partially supported by Polish National Science Centre grant 2016/22/E/ST6/00499.
}}

\author[1]{Marcin Bienkowski}
\author[1]{Artur Kraska}
\author[1]{Pawe{\l} Schmidt}
\affil[1]{Institute of Computer Science, University of Wrocław, Poland}
\authorrunning{M. Bienkowski, A. Kraska and P. Schmidt}
\Copyright{Marcin Bienkowski, Artur Kraska, Pawe{\l} Schmidt}


\ArticleNo{A}
\DOIPrefix{}

\begin{document}

\maketitle

\begin{abstract}
We consider the problem of online Min-cost Perfect Matching with Delays (MPMD)
introduced by Emek et al. (STOC 2016). In this problem, an even number of
requests appear in a metric space at different times and the goal of an
online algorithm is to match them in pairs. In contrast to traditional online
matching problems, in MPMD all requests appear online and an algorithm can match
any pair of requests, but such decision may be delayed (e.g., to find a better
match). The cost is the sum of matching distances and the introduced
delays.

We present the first deterministic online algorithm for this problem. Its competitive
ratio is $O(m^{\log_2 5.5})$ $ = O(m^{2.46})$, where $2 m$ is the number of
requests. This is polynomial in the number of metric space points if all
requests are given at different points. In particular, the bound does not
depend on other parameters of the metric, such as its aspect ratio. Unlike
previous (randomized) solutions for the MPMD problem, our algorithm does not
need to know the metric space in advance.
\end{abstract}

\section{Introduction}

In this paper, we give a deterministic online algorithm for the problem of
Min-cost Perfect Matching with Delays
(MPMD)~\cite{haste-makes-waste,matching-delays}. For an informal description,
imagine that there are human players who are logging in real time into a
gaming website, each wanting to play chess against another human player. The
system pairs the players according to their known capabilities, such as
playing strength. A decision with whom to match a given player can be delayed
until a reasonable match is found. That is, the website tries to
simultaneously minimize two objectives: the waiting times of players and their
dissimilarity, i.e., each player would like to play with another one with
similar capabilities. An algorithm running the website has to work online,
without the knowledge about future player arrivals and make its decision
irrevocably: once two players are paired, they remain paired forever.

\subsection{Problem definition}

More formally, in the MPMD problem there is a metric space $\X$ with a
distance function $\dist: \X \times \X \to \REAL$, both known from the
beginning to an online algorithm. An online part of the input is a sequence of
$2 m$ requests $\{ (p_i,t_i) \}_{i=1}^{2m}$, where point $p_i \in \X$
corresponds to a player in our informal description above and $t_i$ is the
time of its arrival. Clearly, $t_1 \leq t_2 \leq \ldots \leq t_{2m}$. The
integer $m$ is not known a priori to an online algorithm. At any time $\tau$,
an~online algorithm may decide to match any pair of requests $(p_i,t_i)$ and
$(p_j,t_j)$ that have already arrived ($\tau \geq t_i$ and $\tau \geq t_j$) and have not been matched yet. The cost incurred by such \emph{matching edge}
is $\dist(p_i,p_j) + (\tau - t_i) + (\tau - t_j)$, i.e., is the sum of the
\emph{connection cost} and the \emph{waiting costs} of these two requests.

The goal is to eventually match all requests and minimize the total
cost. We use a typical yardstick to measure the performance: a competitive
ratio~\cite{borodin-book}, defined as the maximum, over all inputs, of the
ratios between the cost of an online algorithm and the cost of an~optimal
offline solution \OPT that knows the entire input sequence in advance.

\subsection{Previous work}

The MPMD problem was introduced by Emek et al.~\cite{haste-makes-waste}, who
presented a randomized $O(\log^2 n + \log \Delta)$-competitive algorithm.
There, $n$ is the number of points in the metric space $\X$ and $\Delta$ is
its aspect ratio (the ratio between the largest and the smallest distance
in~$\X$). The competitive ratio was subsequently improved by Azar et
al.~\cite{matching-delays} to $O(\log n)$. They showed that the ratio of any
randomized algorithm is at least $\Omega(\sqrt{\log n})$. The currently best
lower bound of $\Omega(\log n / \log \log n)$ for randomized solutions was
given by Ashlagi et al.~\cite{matching-delays-bipartite-ashlagi}.

So far, the construction of a competitive \emph{deterministic} algorithm for
general metric spaces remained an open problem. It was hypothesized that
competitive ratios achievable by deterministic algorithms might be
superpolynomial in $n$ (cf.~Section 5 of \cite{matching-delays}).
Deterministic algorithms were known only for simple spaces: Azar et
al.~\cite{matching-delays} gave an $O(\textnormal{height})$-competitive
algorithm for trees and Emek et al.~\cite{matching-delays-two-points}
constructed a $3$-competitive deterministic solution for two-point metric (the
competitive ratio is best possible for such metric).

\subsection{Our contribution}

In this paper, we give the first deterministic algorithm for any metric space,
whose competitive ratio is $O(m^{\log_2 5.5}) = O(m^{2.46})$, where $2m$ is the
number of requests. Typically, for our gaming application, $m$ is
smaller than $n$ (although in full generality it can be also larger if
multiple requests arrive at the same point of the metric space $\X$). While
previous solutions to the MPMD problem~\cite{haste-makes-waste,matching-delays}
required $\X$ to be finite and known a priori (to approximate it first by a
random HST tree~\cite{metrics-approximation-optimal} or a random HST tree with
reduced height~\cite{k-server-madry}), our solution works even when $\X$ is
revealed in online manner. That is, we require only that, together with any
request $r$, an online algorithm learns the distances from $r$ to all previous, not yet
matched requests. 

Our online algorithm \ALG uses a simple, local, semi-greedy scheme to find a
suitable matching pair. In the analysis, we fix a final perfect matching of \OPT and
observe what happens when we gradually add matching edges that \ALG creates during its execution.
That is, we trace the evolution of alternating paths and cycles in time. To bound the
cost of \ALG, we charge the cost of an edge that \ALG is adding against the cost
of already existing matching edges from the same alternating path.
Interestingly, our charging argument on alternating cycles bears some
resemblance to the analyses of algorithms for the problems that are not
directly related to MPMD: online metric (bipartite) matching on line
metrics~\cite{online-matching-antoniadis} and offline greedy
matching~\cite{matching-tarjan-reingold}.

\subsection{Related work}
\label{sec:related}

Originally, matching problems have been studied in variants where delaying
decisions was not permitted. The setting most similar to the MPMD problem is
called online \emph{metric bipartite matching}. In involves $m$ \emph{offline
points} given to an algorithm at the beginning and $m$~\emph{requests}
presented in online manner that need to be matched (immediately after their
arrival) to offline points. Both points and requests lie in a common metric
space and the goal is to minimize the weight of a~perfect matching created by
an algorithm. For general metric spaces, the best randomized solution is
$O(\log m)$-competitive~\cite{online-matching-bansal,online-matching-gupta,online-matching-meyerson},
and the deterministic algorithms achieve the optimal competitive ratio of
$2m-1$~\cite{online-matching-pruhs,online-matching-khuller}. Interestingly,
even for line metrics~\cite{online-matching-antoniadis,online-matching-line-lower-bound,online-matching-koutsoupias},
the best known deterministic algorithm attains a competitive ratio that is
polynomial in $m$~\cite{online-matching-antoniadis}.

In comparison, in the MPMD problem considered in this paper, all $2m$ requests appear in
online manner, $m$ is not known to an algorithm, and we allow to match any pair of them. 
That said, there is
also a bipartite variant of the MPMD problem, in which all requests
appear online, but $m$ of them are negative and $m$ are positive. An algorithm
may then only match pairs of requests of different
polarities~\cite{matching-delays-bipartite-azar,matching-delays-bipartite-ashlagi}.

The MPMD problem can be cast as augmenting min-cost perfect matching with a
time axis, allowing the algorithm to delay its decisions, but penalizing the
delays. There are many other problems that use this paradigm: most notably the
ski-rental problem and its continuous counterpart, the spin-block
problem~\cite{snoopy-caching-randomized}, where a purchase decision can be
delayed until renting cost becomes sufficiently large. Such rent-or-buy (wait-or-act)
trade-offs are also found in other areas, for example in aggregating messages
in computer networks~\cite{tcp-ack-albers,
aggregation-wads,
tcp-ack-det-journal,
tcp-ack,
khanna-message-aggregation,
delay-sensitive-aggregation},
in aggregating orders in supply-chain 
management~\cite{multilevel-aggregation,
jrp-soda,
aggregation-bkv,
multilevel-aggregation-buchbinder,
jrp-online-buchbinder,
aggregation-survey}
or in some scheduling variants~\cite{make-to-order-scheduling}.

Finally, there is a vast amount of work devoted to other online matching variant, 
where offline points and online requests are connected by graph edges 
and the goal is to maximize the weight or the cardinality of the produced matching. 
These types of matching problems have been studied since the seminal work of Karp et~al.~\cite{online-matching}
and are motivated by applications to online auctions~\cite{online-matching-simple,
adwords-primal-dual,
concave-matching,
ranking-primal-dual,
online-matching,
bipartite-matching-strongly-lp,
adwords-lp,
adwords-ec}. They were also studied under stochastic assumptions on the input, see, e.g., a survey by Mehta~\cite{adwords-survey}.

\section{Algorithm}

We will identify requests with the points at which they arrive. To this end,
we assume that all requested points are different, but we allow distances
between different metric points to be zero. For any request $p$, we denote the
time of its arrival by $\atime(p)$.

Our algorithm is parameterized with real numbers $\alpha > 0$ and
$\beta > 1$, whose exact values will be optimized later. 
For any request $p$, we define its waiting time at time $\tau \geq \atime(p)$ as
\[
	\wait_\tau(p) = \tau - \atime(p) 
\]
and its budget at time $\tau$ as
\[
	\budget_\tau(p) = \alpha \cdot \wait_\tau(p)
	\enspace.
\]

Our online algorithm \ALG matches two requests $p$ and $q$ at time $\tau$ as soon as the following
two conditions are satisfied.
\begin{itemize}
\item \emph{Budget sufficiency}: $\budget_\tau(p) + \budget_\tau(q) \geq \dist(p,q)$.
\item \emph{Budget balance}: $\budget_\tau(p) \leq \beta \cdot \budget_\tau(q)$ and 
$\budget_\tau(q) \leq \beta \cdot \budget_\tau(p)$. 
\end{itemize}

Note that the budget balance condition is equivalent to relations on waiting 
times, i.e.,
$\wait_\tau(p) \leq \beta \cdot \wait_\tau(q)$ and 
$\wait_\tau(q) \leq \beta \cdot \wait_\tau(p)$. 

If the conditions above are met simultaneously for many point pairs,
we break ties arbitrarily, and process them in any order. Note that at the
time when $p$ and $q$ become matched, the sum of their budgets may exceed
$\dist(p,q)$. For example, this occurs when $q$ appears at time strictly
larger than $\atime(p) + \dist(p,q)$: they are then matched by \ALG as soon
as the budget balance condition becomes true.

The observation below follows immediately by the definition of \ALG.

\begin{observation}
\label{obs:free_pair}
Fix time $\tau$ and two requests $p$ and $q$, such that $\atime(p) \leq \tau$
and $\atime(q) \leq \tau$. Assume that neither $p$ nor $q$ has been
matched by \ALG strictly before time $\tau$. Then exactly one of the following conditions 
holds:
\begin{itemize}
	\item $\alpha \cdot (\wait_{\tau}(p) + \wait_{\tau}(q)) \leq \dist(p, q)$,
	\item $\alpha \cdot (\wait_{\tau}(p) + \wait_{\tau}(q)) > \dist(p, q)$ and 
		$\wait_{\tau}(p) \geq \beta \cdot \wait_{\tau}(q)$,
	\item $\alpha \cdot (\wait_{\tau}(p) + \wait_{\tau}(q)) > \dist(p, q)$ and 
		$\wait_{\tau}(q) \geq \beta \cdot \wait_{\tau}(p)$.
\end{itemize}
\end{observation}

\section{Analysis}

To analyze the performance of \ALG, we look at matchings generated by \ALG
and by an~optimal offline algorithm \OPT. If points $p$ and $q$ were matched
at time~$\tau$ by \ALG, then we say that \ALG creates a (matching) edge $e = (p,q)$.
Its cost is 
\[
	\costALG(e) = \costALG(p,q) = \dist(p,q) + \wait_\tau(p) + \wait_\tau(q)
	\enspace.
\]
We call $e$ an \ALG-edge. The $\costOPT$ of an edge in the solution of \OPT (an
\OPT-edge) is defined analogously. In an optimal solution, however,
the matching time is always equal to the arrival time of the later of two
matched requests.

We now consider a dynamically changing graph consisting of requested points,
\OPT-edges and \ALG-edges. For the analysis, we assume that it changes in the
following way: all requested points and all \OPT-edges are present in the
graph from the beginning, but the \ALG-edges are added to the graph in $m$
steps, in the order they are created by \ALG.

At all times, the matching edges present in the graph form alternating paths
or cycles (i.e., paths or cycles whose edges are interleaved \ALG-edges and
\OPT-edges). Furthermore, any node-maximal alternating path
starts and ends with \OPT-edges. Assume now that a~matching edge~$e$ created
by \ALG is added to the graph. It may either connect the ends of two different
alternating paths, thus creating a single longer alternating path or connect
the ends of one alternating path, generating an alternating cycle. In the
former case, we call edge~$e$ \emph{non-final}, in the latter case ---
\emph{final}. Note that at the end of the \ALG execution, when $m$
\ALG-edges are added, the graph contains only alternating cycles.

We extend the notion of cost to alternating path and cycles. 
For any cycle $C$, $\cost(C)$ is simply the 
sum of costs of its edges: the cost of 
an \OPT-edge on such cycle is the cost paid by \OPT and the cost of an \ALG-edge
is that of \ALG. We also define $\costOPT(C)$, $\costALG(C)$ and $\costALGNF(C)$ 
as the costs of \OPT-edges, \ALG-edges and non-final \ALG-edges on cycle $C$, respectively.
Clearly, $\costALG(C) + \costOPT(C) = \cost(C)$.
We define the same notions for alternating paths; as a path $P$
does not contain final \ALG-edges, $\costALGNF(P) = \costALG(P)$.

An alternating path is called \emph{$\kappa$-step maximal alternating path} if it
exists in the graph after \ALG matched $\kappa$ pairs 
and it cannot be extended, i.e., it ends with two requests
that are not yet matched by the first $\kappa$ \ALG-edges.

\subsection{Tree construction}
\label{sec:tree_def}

To facilitate the analysis, along with the graph, we create a dynamically
changing forest $F$ of binary trees, where each leaf of~$F$ corresponds to an
\OPT-edge and each internal (non-leaf) node of $F$ to a non-final \ALG-edge (and vice
versa). After \ALG matched $\kappa$ pairs, each subtree of $F$ corresponds to
a~$\kappa$-step maximal alternating path or to an alternating cycle. More
precisely, at the beginning, $F$~consists of $m$ single nodes representing
\OPT-edges. Afterwards, whenever an \ALG-edge is created, we perform the following operation on $F$.
\begin{itemize}
\item
When a non-final \ALG-edge $e = (p,q)$ is added to the graph, we look at
the two alternating paths $P$ and $Q$ that end with $p$ and $q$,
respectively. We take the corresponding trees $T(P)$ and $T(Q)$ of $F$. We add a
node $v(e)$ (representing edge $e$) to $F$ and make $T(P)$ and $T(Q)$ its
subtrees.
\item
When a final \ALG-edge $e = (p,q)$ is added to the graph, it turns an
alternating path~$P$ into an alternating cycle $C$. We then simply say that
the tree $T(P)$ that corresponded to~$P$, now corresponds to $C$.
\end{itemize}
An example of the graph and the associated forest $F$ is presented in \lref[Figure]{fig:forest_creation}.

\begin{figure}
\centering
\includegraphics[width=0.95\textwidth]{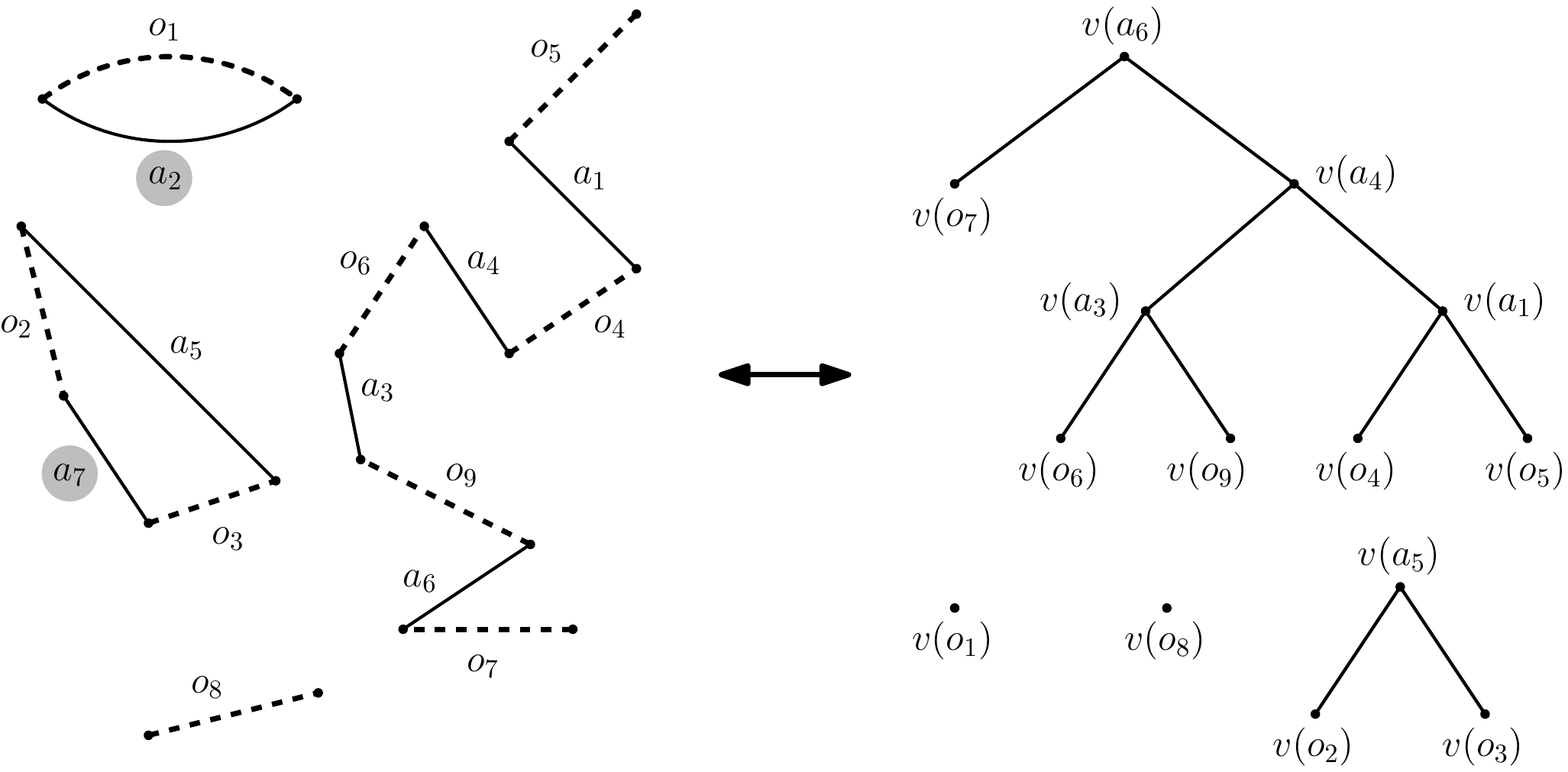}
\caption{%
The left side contains an example graph consisting of all \OPT-edges $o_1,
o_2, \ldots, o_9$ (dashed lines) and the first $\kappa = 7$
\ALG-edges $a_1, a_2, \ldots, a_7$ (solid lines). 
\ALG-edges are numbered in the order they were created and added to the graph. 
Shaded \ALG-edges ($a_2$ and $a_7$) are final, the remaining ones are
non-final. \\ The right side depicts the corresponding forest $F$: leaves of
$F$ represent \OPT-edges and non-leaf nodes of $F$ correspond to non-final
\ALG-edges. Trees rooted at nodes $v(o_1)$ and $v(a_5)$ represent alternating cycles and those rooted at nodes $v(a_6)$ and $v(o_8)$ represent alternating paths in the
graph.
}
\label{fig:forest_creation}
\end{figure}

For any tree node $w$, we define its weight $\weight(w)$ as the cost of the
corresponding matching edge, i.e., the cost of an \OPT-edge for a leaf and the
cost of a non-final \ALG-edge for a non-leaf node. For any node $w$, by $T_w$
we denote the tree rooted at $w$. We extend the notion of weight in a natural
manner to all subtrees of $F$. In these terms, the weight of a~tree~$T$ in $F$ is
equal to the total cost of the corresponding alternating path. (If $T$ 
represents an alternating cycle $C$, then its weight is equal to the cost 
of $C$ minus the cost of the final \ALG-edge from $C$.)

Note that we consistently used terms ``points'' and ``edges'' for objects that
\ALG and \OPT are operating on in the metric space $\X$. On the other hand,
the term ``nodes'' will always refer to tree nodes in $F$ and we will not use
the term ``edge'' to denote an edge in $F$.

\subsection{Outline of the analysis}

Our approach to bounding the cost of
\ALG is now as follows. We look at the forest $F$ at the end of \ALG execution.
The corresponding graph contains only alternating cycles. The cost of
non-final \ALG-edges is then, by the definition, equal to the total weight of
internal (non-leaf) nodes of $F$, while the cost of \OPT-edges is equal to the total
weight of leaves of $F$. Hence, our goal is to relate the total weight of any
tree to the weight of its leaves.

The central piece of our analysis is showing that for any internal node $w$
with children $u$ and $v$, it holds that $\weight(w) \leq \xi \cdot
\min\{ \weight(T_u), \weight(T_v)\}$, where $\xi$ is a constant depending on 
parameters $\alpha$ and $\beta$ (see \lref[Corollary]{cor:tree_bound}). Using this
relation, we will bound the total weight of any tree by $O(m^{\log_2 {(\xi+2) -
1}})$ times the total weight of its leaves. This implies the same bound on the
ratio between non-final \ALG-edges and \OPT-edges on each alternating cycle.

Finally, we show that the cost of final \ALG-edges incurs at most an
additional constant factor in the total cost of \ALG.

\subsection{Cost of non-final ALG-edges}
\label{sec:tree_node_weights}

As described in \lref[Section]{sec:tree_def}, when \ALG adds a $\kappa$-th
\ALG-edge $e$ to the graph, and this edge is non-final, $e$ joins two
$(\kappa-1)$-step maximal alternating paths $P$ and $Q$. We will bound
$\costALG(e)$ by a constant (depending on $\alpha$ and $\beta$) times $\min \{
\cost(P), \cost(Q) \}$. We start with bounding the waiting cost of \ALG
related to one endpoint of $e$.

\begin{lemma}
\label{lem:crucial_lemma}
	Let $e = (p,q)$ be the $\kappa$-th \ALG-edge added at time $\tau$, such
	that $e$ is non-final. Let $P = (a_1, a_2, \ldots, a_\ell$) be the
	$(\kappa-1)$-step maximal alternating path ending at $p = a_1$. Then,
	$\wait_\tau(p) \leq \max \{ \alpha^{-1}, {\beta}/{(\beta -1)} \} \cdot
	\cost(P)$.
\end{lemma}

\begin{proof}
First we lower-bound the cost of an alternating path $P$. We look at any edge
$(a_i,a_{i+1})$ from $P$. Its cost (no matter whether paid by \ALG or \OPT) is
certainly larger than $\dist(a_i, a_{i+1}) + |\atime(a_i) - \atime(a_{i+1})|$.
Therefore, using triangle inequality (on distances and times), we obtain
\begin{align}
	\cost(P) 
	\geq &\; \sum_{i=1}^{\ell-1} \left( \dist(a_i,a_{i+1}) + | \atime(a_i) - \atime(a_{i+1}) | \right) \nonumber \\
	\geq &\; \dist(a_1,a_\ell) + |\atime(a_1) - \atime(a_\ell)| 
	\enspace.
	\label{eq:triangle_inequality}
\end{align}
Therefore, in our proof we will simply bound $\wait_\tau(p) = \wait_\tau(a_1)$ 
using either $\dist(a_1,a_\ell)$ or $|\atime(a_1) - \atime(a_\ell)|$.

Recall that $\ALG$ matches $a_1$ at time $\tau$. Consider the state of
$a_\ell$ at time $\tau$. If $a_\ell$ has not been presented to \ALG yet
($\atime(a_\ell) > \tau$), then $\wait_\tau(a_1) = \tau -
\atime(a_1) < \atime(a_\ell) - \atime(a_1) < \beta / (\beta-1) 
\cdot ( \atime(a_\ell) - \atime(a_1) )$,  and the lemma follows.

In the remaining part of the proof, we assume that $a_\ell$ was already
presented to the algorithm ($\atime(a_\ell) \leq \tau$). As $P$ is a
$(\kappa-1)$-step maximal alternating path, $a_\ell$ is not matched by $\ALG$
right after $\ALG$ creates $(\kappa-1)$-th matching edge. The earliest time
when $a_\ell$ may become matched is when \ALG creates the next, $\kappa$-th
matching edge, i.e., at time $\tau$. Therefore $a_\ell$ is not matched before
time $\tau$.

Now observe that there must be a reason for which 
requests $a_1$ and $a_\ell$ have not been matched with each other before time $\tau$. 
Roughly speaking, either the sum of budgets of requests $a_1$ and $a_\ell$ does not suffice to cover the cost of
$\dist(a_1,a_\ell)$ or one of them waits significantly longer than the other.
Formally, we apply \lref[Observation]{obs:free_pair} to pair $(a_1,a_\ell)$ obtaining three possible cases.
In each of the cases we bound $\wait_\tau(a_1)$ appropriately.

\begin{description}
	\item[Case 1 (insufficient budgets).]
	If $\alpha \cdot (\wait_\tau(a_1) + \wait_\tau(a_\ell)) \leq \dist(a_1, a_\ell)$, then 
	by non-negativity of $\wait_\tau(a_\ell)$, it follows that 
	$\wait_\tau(a_1) \leq \alpha^{-1} \cdot \dist(a_1,a_\ell)$. \\	

	\item[Case 2 ($a_1$ waited much longer than $a_\ell$).] 
	If $\alpha \cdot (\wait_\tau(a_1) + \wait_\tau(a_\ell)) >  \dist(a_1, a_\ell)$ 
	and $\wait_\tau(a_1) \geq \beta \cdot \wait_\tau(a_\ell)$, then 
	$\atime(a_\ell) - \atime(a_1) = \wait_\tau(a_1) - \wait_\tau(a_\ell) \geq (1-1/\beta)\cdot \wait_\tau(a_1)$.
	Therefore, 
	$\wait_\tau(a_1) \leq \beta / (\beta -1) \cdot |\atime(a_1) - \atime(a_\ell)|$. \\

	\item[Case 3 ($a_\ell$ waited much longer than $a_1$).] 
	If $\alpha \cdot (\wait_\tau(a_1) + \wait_\tau(a_\ell)) > \dist(a_1, a_\ell)$ 
	and $\wait_\tau(a_\ell) \geq \beta \cdot \wait_\tau(a_1)$, then
	$\atime(a_1) - \atime(a_\ell) = \wait_\tau(a_\ell) - \wait_\tau(a_1) \geq (\beta - 1) \cdot \wait_\tau(a_1)$. 
	Thus, 
	$\wait_\tau(a_1) \leq 1/(\beta-1) \cdot |\atime(a_1) - \atime(a_\ell)|
	< \beta/(\beta-1) \cdot |\atime(a_1) - \atime(a_\ell)|$.
	\qedhere
\end{description}	
\end{proof}

\begin{lemma}
\label{lem:crucial_lemma_2}
	Let $e = (p,q)$ be the $\kappa$-th \ALG-edge, such that $e$ is non-final. Let $P = (a_1, a_2, \ldots, a_\ell)$ 
	and $Q = (b_1, b_2, \ldots, b_{\ell'})$ be the $(\kappa-1)$-step maximal alternating path ending at $p = a_1$
	and $q = b_1$, respectively. Then, 
	\[
		\costALG(e) \leq 
			(1+\alpha) \cdot (\beta + 1) \cdot \max \{ \alpha^{-1}, \beta/ (\beta -1) \}
			\cdot \min \{ \cost(P), \cost(Q) \}
		\enspace.
	\]
\end{lemma}

\begin{proof}
	Let $\tau$ be the time when $p$ is matched with $q$ by $\ALG$. 
	Using the definition of $\costALG$, we obtain
	\begin{align}
		 \costALG(p,q) 
			& =	\dist(p, q) + \wait_\tau(p) + \wait_\tau(q) \nonumber \\
			& \leq \budget_\tau(p) + \budget_\tau(q) + \wait_\tau(p) + \wait_\tau(q) \nonumber \\
			& = (1+\alpha) \cdot (\wait_\tau(p) + \wait_\tau(q)) \nonumber \\
			& \leq (1+\alpha) \cdot (\beta+1) \cdot \min \{ \wait_\tau(p), \wait_\tau(q) \}
			\enspace. 
			\label{eq:path_bound}
	\end{align}
	The first inequality follows by the budget sufficiency condition of \ALG and the second one by the budget balance condition.

	By \lref[Lemma]{lem:crucial_lemma}, 
	$\wait_\tau(p) \leq \max \{ \alpha^{-1}, {\beta}/{(\beta -1)} \} \cdot \cost(P)$
	and 
	$\wait_\tau(q) \leq \max \{ \alpha^{-1},$ $ {\beta}/{(\beta -1)} \} \cdot \cost(Q)$,
	which combined with \eqref{eq:path_bound} immediately yield the lemma.
\end{proof}

Recall now the iterative construction of forest $F$ from
\lref[Section]{sec:tree_def}: whenever a non-final matching edge $e$ created
by \ALG joins two alternating paths $P$ and $Q$, we add a~new node $w$ to $F$,
such that $\weight(w) = \costALG(e)$ and make trees $T(P)$ and $T(Q)$ its
children. These trees correspond to paths $P$ and $Q$, and satisfy
$\weight(T(P)) = \cost(P)$ and $\weight(T(Q)) = \cost(Q)$. Therefore,
\lref[Lemma]{lem:crucial_lemma_2} immediately implies the following equivalent
relation on tree weights.

\begin{corollary}
\label{cor:tree_bound}
Let $w$ be an internal node of the forest $F$ whose children are $u$ and $v$. 
Then, $\weight(w) \leq (1+\alpha) \cdot (\beta + 1) \cdot 
\max \{ \alpha^{-1}, \beta/ (\beta -1) \} \cdot \min \{ \weight(T_u), \weight(T_v)\}$.
\end{corollary}

This relation can be used to express the total weight of a tree of $F$ in
terms of the total weight of its leaves. The proof of the following technical lemma is
deferred to \lref[Section]{sec:total_tree_cost}. Here, we present how to use
it to bound the cost of \ALG on non-final edges of a single alternating cycle.

\begin{lemma}
\label{lem:tree_cost}
Let $T$ be a weighted full binary tree and $\xi \geq 0$ be any constant. Assume that for each internal node $w$ with children $u$ and $v$, their weights satisfy $\weight(w) \leq \xi \cdot \min \{ \weight(T_u), \weight(T_v) \}$.
Then, 
\[
	\weight(T) \leq (\xi + 2) \cdot |\leaves(T)|^{\log_2(\xi/2+1)} \cdot \weight(\leaves(T))
	\enspace,
\] 
where $\leaves(T)$ is the set of leaves of $T$ and $\weight(\leaves(T))$ is their total weight.
\end{lemma}

\begin{lemma}
	\label{lem:non_final_bound}
	Let $C$ be an alternating cycle obtained from combining matchings of \ALG
	and \OPT. Then
	$\costALGNF(C) \leq (\xi + 2) \cdot m^{\log_2(\xi/2+1)} \cdot \costOPT(C)$, where
	$\xi = (1+\alpha) \cdot (\beta + 1) \cdot
\max \{ \alpha^{-1}, \beta/ (\beta -1) \}$.
\end{lemma}

\begin{proof}
	As described in \lref[Section]{sec:tree_def}, $C$ is associated with a tree $T$ 
	from forest $F$, such that \OPT-edges of $C$ correspond to the set of 
	leaves of $T$ (denoted $L(T)$) and non-final \ALG-edges of $C$ correspond 
	to internal (non-leaf) nodes of $T$. Hence, 
	$\costOPT(C) = \weight(\leaves(T))$ and 
	$\costALGNF(C) + \costOPT(C) = \weight(T)$.
		
	By \lref[Corollary]{cor:tree_bound}, the weight of any internal tree node $w$ with
	children $u, v$ satisfies
	$\weight(w) \leq \xi \cdot \min \{ \weight(T_u), \weight(T_v)\}$.
	Therefore, we may apply \lref[Lemma]{lem:tree_cost} to tree $T$, obtaining
	$\weight(T) \leq (\xi + 2) \cdot |\leaves(T)|^{\log_2(\xi/2+1)} \cdot
	\weight(L(T))$, and thus
	\begin{align*}
		\costALGNF(C) 
		\leq \weight(T)
		\leq &\; (\xi + 2) \cdot |L(T)|^{\log_2(\xi/2+1)} \cdot \weight(L(T)) \\
		\leq &\; (\xi + 2) \cdot m^{\log_2(\xi/2+1)} \cdot \weight(L(T)) \\
		= &\; (\xi + 2) \cdot m^{\log_2(\xi/2+1)} \cdot \costOPT(C)
		\enspace.
	\end{align*}
	The last inequality follows as $|\leaves(T)|$, the number of $T$ leaves, 
	is equal to the number of \OPT-edges on cycle $C$, which is clearly at most $m$.
\end{proof}


\subsection{Cost of final ALG-edges}

In the previous section, we derived a bound on the cost of all non-final
\ALG-edges. The following lemma shows that the cost of final \ALG-edges 
contribute at most a constant factor to the competitive ratio.

\begin{lemma}
\label{lem:final_edge_bound}
	Let $e$ be a final \ALG-edge matched at time $\tau$ and $C$ be the
	alternating cycle containing $e$. Then $\costALG(e) \leq (1+\alpha) \cdot
	\max \{ \alpha^{-1},(\beta+1)/(\beta-1) \} \cdot (\costALGNF(C) + \costOPT(C))$.
\end{lemma}
\begin{proof}
	Fix a final \ALG-edge $e = (p, q)$, where $\atime(q) \geq \atime(p)$.
	By the budget sufficiency condition of \ALG, 
	\begin{equation}
		\label{eq:final_edges_1}
		\costALG(e) \leq (1+\alpha) \cdot(\wait_\tau(p) + \wait_\tau(q))
		\enspace.
	\end{equation}

	Our goal now is to bound $\wait_\tau(p) + \wait_\tau(q)$ in terms of
	$\dist(p,q)$ or $\atime(q) - \atime(p)$. Observe that whenever \ALG
	matches two requests, the budget sufficiency condition of \ALG or one of
	the inequalities of the budget balance condition is satisfied with
	equality. We apply this observation to pair $(p,q)$.
	
	\begin{itemize}
	\item If the budget sufficiency condition holds with equality, 
    $\alpha \cdot (\wait_\tau(p) + \wait_\tau(q)) = \dist(p, q)$, and 
    therefore $\wait_\tau(p) + \wait_\tau(q) = \alpha^{-1} \cdot \dist(p, q)$.

    \item 
    If the budget balance condition holds with equality, 
   	$\beta \cdot \wait_\tau(q) = \wait_\tau(p)$. 
   	Then,
   	\begin{align*}
   		(\beta - 1) \cdot (\wait_\tau(p) + \wait_\tau(q)) 
   		= &\; (\beta - 1) \cdot (\beta + 1) \cdot \wait_\tau(q) \\
   		= &\; (\beta + 1) \cdot (\wait_\tau(p) - \wait_\tau(q)) \\
   		= &\; (\beta + 1) \cdot (\atime(q) - \atime(p))
   		\enspace.
   	\end{align*}
	\end{itemize}
	Hence, in either case it holds that 
	\begin{equation}
	\label{eq:final_edges_2}
	\wait_\tau(p) + \wait_\tau(q) \leq \max \left\{ \alpha^{-1}, \frac{\beta+1}{\beta-1} 
		\right\} \cdot (\dist(p, q) + |\atime(q) - \atime(p)|)
	\enspace.
	\end{equation}

	Finally, we bound $\dist(p, q) + |\atime(q) - \atime(p)|$ in terms of
	costs of other edges of $C$. These edges form a path $P =
	(a_1,a_2,\ldots,a_\ell)$, where $a_1 = p$ and $a_\ell = q$. By the
	triangle inequality applied to distances and time differences (in the same
	way as in~\eqref{eq:triangle_inequality}), we obtain that
	\begin{equation}
	\label{eq:final_edges_3}
		\dist(p, q) + |\atime(q) - \atime(p)| \leq \cost(P) = \costALGNF(C) + \costOPT(C)
		\enspace.
	\end{equation}
	The lemma follows immediately by combining 
	\eqref{eq:final_edges_1}, \eqref{eq:final_edges_2} and \eqref{eq:final_edges_3}.
\end{proof}

\subsection{The competitive ratio}

Finally, we optimize constants $\alpha$ and $\beta$ used throughout the previous
sections and bound the competitiveness of \ALG.

\begin{theorem}
	For $\beta = 2$ and $\alpha = 1/2$, the competitive ratio of \ALG is
	$O(m^{\log_25.5})=O(m^{2.46})$, where $2m$ is the number of requests in
	the input sequence.
\end{theorem}

\begin{proof}
	The union of matchings constructed by \ALG and \OPT can be split into a
	set $\C$ of disjoint cycles. It is sufficient to show that we have the
	desired performance guarantee on each cycle from $\C$.

	Fix a cycle $C\in\C$. Let $e=(p,q)$ be the final \ALG-edge of $C$. By
	\lref[Lemma]{lem:final_edge_bound}, $\costALG(e) \leq 4.5 \cdot
	\left(\costALGNF(C) + \costOPT(C) \right)$. Therefore, the competitive
	ratio of \ALG is at most
	\[
		\frac{\costALG(C)}{\costOPT(C)} 
			\leq \frac{ 5.5\cdot \costALGNF(C) + 4.5\cdot \costOPT(C)}{\costOPT(C)}
			\leq O(m^{\log_2{5.5}}) = O(m^{2.46})\enspace,
	\]
	where the second inequality follows by \lref[Lemma]{lem:non_final_bound}.
\end{proof}

\section{Relating weights in trees (proof of Lemma \ref*{lem:tree_cost})}
\label{sec:total_tree_cost}

We start with the following technical claim that will facilitate 
the inductive proof of \lref[Lemma]{lem:tree_cost}.

\begin{lemma}
\label{lem:convex_ineq}
	Fix any constant $\xi \geq 0$ and let $f(a) = a^{\log_2(\xi+2)}$.
	Then, $\xi \cdot \min \{ f(x), f(y) \} + f(x) + f(y) \leq f(x+y)$ for all $x, y \geq 0$.
\end{lemma}

\begin{proof}
	Fix any $z \geq 0$ and let $g_z(a) = (\xi+1) \cdot f(a) + f(z-a)$. 
	We observe that $g_z(0) = f(z)$ and $g_z(z/2) = (\xi+1) \cdot f(z/2) + f(z/2) 
		= (\xi+2) \cdot (z/2)^{\log_2(\xi+2)} 
		= z^{\log_2(\xi+2)} = f(z)$.
	Moreover, the function $g_z$ is convex as it is a sum of two convex functions.
	As $g_z(0) = g_z(z/2) = f(z)$, by convexity, $g_z(a) \leq f(z)$ for any $a \in [0,z/2]$.

	To prove the lemma, assume without loss of generality that $x\leq y$. By the monotonicity,
	$f(x) \leq f(y)$, and therefore
	\begin{align*}
		\xi \cdot \min \{ f(x), f(y) \} + f(x) + f(y) 
			= &\; (\xi + 1) \cdot f(x) + f((x+y)-x) \\
			= &\; g_{x+y}(x) \\
			\leq &\; f(x+y)\enspace.
	\end{align*}
	The last inequality follows as $x \leq (x+y)/2$. 
\end{proof}

\begin{proof}[Proof of \text{\lref[Lemma]{lem:tree_cost}}]
We scale weights of all nodes, so that the average weight of each leaf is~$1$,
i.e., we define a scaled weight function $\weightt$ as  
\[
	\weightt(w) = \weight(w) \cdot \frac{|\leaves(T)|}{\weight(\leaves(T))} 
	\enspace.
\]
Note that $\weightt$ also satisfies $\weightt(w) \leq \xi \cdot \min \{
\weightt(T_u), \weightt(T_v) \}$. Moreover, since we scaled all weighs in the
very same way, ${\weightt(T)}/{\weightt(\leaves(T))} =
{\weight(T)}/{\weight(\leaves(T))}$, and hence to show the lemma, it suffices to bound 
the term $\weightt(T) /\weightt(\leaves(T))$.

For any node $w \in T$ and the corresponding subtree $T_w$ rooted at $w$, we define 
$\size(T_w) = \weightt(\leaves(T_w)) + |\leaves(T_w)|$. 
We inductively show that for any node of $w \in T$, it holds that 
\begin{equation}
\label{eq:induction}
	\weightt(T_w) \leq \size(T_w)^{\log_2(\xi+2)}
	\enspace.
\end{equation}
For the induction basis, assume that 
$w$ is a leaf of $T$. Then,
\[
	\weightt(T_w) = \weightt(\leaves(T_w)) \leq \size(T_w) \leq \size(T_w)^{\log_2(\xi+2)}\enspace,
\] 
where the last inequality follows as  $\size(T_w) \geq |L(T_w)| = 1$ and $\xi >0$.

For the inductive step, let $w$ be a non-leaf node of $T$ and let
$u$ and $v$ be its children. Then,
	\begin{align*}
	\weightt(T_w) &= \weightt(T_u) + \weightt(T_v) + \weightt(w)\\
				 &\leq \weightt(T_u) + \weightt(T_v) +\xi \cdot \min \,\{\, \weightt(T_u),\, \weightt(T_v) \,\} \\
				 &\leq \size(T_u)^{\log_2(\xi+2)} + \size(T_v)^{\log_2(\xi+2)} 
				 	+ \xi \cdot  \min \{
					 	\size(T_u)^{\log_2(\xi+2)},\, \size(T_v)^{\log_2(\xi+2)} \}  \\
				 &\leq (\size(T_u)+\size(T_v))^{\log_2(\xi+2)}\\
				 &=\size(T_w)^{\log_2(\xi+2)}\enspace.
	\end{align*}
The first inequality follows by the lemma assumption and the second one by the
inductive assumptions for $T_u$ and $T_v$. The last inequality is a
consequence of \lref[Lemma]{lem:convex_ineq} and the final equality follows
by the additivity of function $\size$.

Recall that we scaled weights so that $\weightt(\leaves(T)) = |\leaves(T)|$. Therefore, applying
\eqref{eq:induction} to the whole tree $T$ yields
$\weightt(T) \leq (\weightt(\leaves(T)) + |\leaves(T)|)^{\log_2(\xi+2)} = (2\cdot |\leaves(T)|)^{\log_2(\xi+2)} = (\xi+2) \cdot |\leaves(T)|^{\log_2(\xi+2)}$. Hence,
\[
	\frac{\weight(T)}{\weight(\leaves(T))} 
	= \frac{\weightt(T)}{\weightt(\leaves(T))} 
	\leq \frac{(\xi+2) \cdot |\leaves(T)|^{\log_2(\xi+2)}}{|\leaves(T)|} 
	= (\xi+2) \cdot |\leaves(T)|^{\log_2(\xi/2+1)}\enspace ,
\]
which concludes the proof.
\end{proof}
\section{Conclusions}

We showed a deterministic algorithm \ALG for the MPMD problem whose
competitive ratio is $O(m^{\log_2 5.5})$. The currently best lower bound
(holding even for randomized solutions) is $\Omega(\log n / \log \log
n)$~\cite{matching-delays-bipartite-ashlagi}. A natural research direction
would be to narrow this gap.

It is not known whether the analysis of our algorithm is tight. However,
one can show that its
competitive ratio is at least $\Omega(m^{\log_2 1.5}) = \Omega(m^{0.58})$. To
this end, assume that all requests arrive at the same time. For such
input, \OPT does not pay for delays and simply returns the min-cost perfect matching. On the other hand,
\ALG computes the same matching as a~greedy routine (i.e., it greedily connects two nearest, not yet
matched requests). Hence, even if we neglect the delay costs of \ALG, its competitive
ratio would be at least the approximation ratio of the greedy algorithm for min-cost
perfect matching. The latter was shown to be $\Theta(m^{\log_2 1.5})$
by Reingold and Tarjan~\cite{matching-tarjan-reingold}.

The reasoning above indicates an inherent difficulty of the problem. In order to 
beat the $\Omega(m^{\log_2 1.5})$ barrier, an online algorithm has to 
handle settings when all requests are given simultaneously more effectively. 
In particular, for such and similar input instances it has to employ a non-local and
non-greedy policy of choosing requests to match.

\bibliographystyle{plainurl}
\bibliography{references}

\begin{thebibliography}{10}

\bibitem{tcp-ack-albers}
Susanne Albers and Helge Bals.
\newblock Dynamic {TCP} acknowledgment: Penalizing long delays.
\newblock {\em SIAM Journal on Discrete Mathematics}, 19(4):938--951, 2005.

\bibitem{online-matching-antoniadis}
Antonios Antoniadis, Neal Barcelo, Michael Nugent, Kirk Pruhs, and Michele
  Scquizzato.
\newblock A o(n)-competitive deterministic algorithm for online matching on a
  line.
\newblock In {\em Proc. 12th Workshop on Approximation and Online Algorithms
  (WAOA)}, pages 11--22, 2014.

\bibitem{matching-delays-bipartite-ashlagi}
Itai Ashlagi, Yossi Azar, Moses Charikar, Ashish Chiplunkar, Ofir Geri, Haim
  Kaplan, Rahul Makhijani, Yuyi Wang, and Roger Wattenhofer.
\newblock Min-cost bipartite perfect matching with delays.
\newblock 2017.
\newblock URL: \url{https://web.stanford.edu/~iashlagi/papers/mbpmd.pdf}.

\bibitem{matching-delays-bipartite-azar}
Yossi Azar, Ashish Chiplunkar, and Haim Kaplan.
\newblock Polylogarithmic bounds on the competitiveness of min-cost (bipartite)
  perfect matching with delays.
\newblock 2016.
\newblock URL: \url{https://arxiv.org/abs/1610.05155}.

\bibitem{matching-delays}
Yossi Azar, Ashish Chiplunkar, and Haim Kaplan.
\newblock Polylogarithmic bounds on the competitiveness of min-cost perfect
  matching with delays.
\newblock In {\em Proc. 28th ACM-SIAM Symp. on Discrete Algorithms (SODA)},
  pages 1051--1061, 2017.

\bibitem{make-to-order-scheduling}
Yossi Azar, Amir Epstein, {\L}ukasz Je{\.z}, and Adi Vardi.
\newblock Make-to-order integrated scheduling and distribution.
\newblock In {\em Proc. 27th ACM-SIAM Symp. on Discrete Algorithms (SODA)},
  pages 140--154, 2016.

\bibitem{online-matching-bansal}
Nikhil Bansal, Niv Buchbinder, Anupam Gupta, and Joseph Naor.
\newblock A randomized {$O(\log^2 k)$}-competitive algorithm for metric
  bipartite matching.
\newblock {\em Algorithmica}, 68(2):390--403, 2014.

\bibitem{k-server-madry}
Nikhil Bansal, Niv Buchbinder, Aleksander M\k{a}dry, and Joseph Naor.
\newblock A polylogarithmic-competitive algorithm for the \emph{k}-server
  problem.
\newblock {\em Journal of the ACM}, 62(5):40:1--40:49, 2015.

\bibitem{multilevel-aggregation}
Marcin Bienkowski, Martin B{\"o}hm, Jaroslaw Byrka, Marek Chrobak, Christoph
  D{\"u}rr, Luk\'{a}\v{s} Folwarczn\'{y}, {\L}ukasz Je\.{z}, Ji\v{r}\'{\i}
  Sgall, Nguyen~Kim Thang, and Pavel Vesel\'{y}.
\newblock Online algorithms for multi-level aggregation.
\newblock In {\em Proc. 24th European Symp. on Algorithms (ESA)}, pages
  12:1--12:17, 2016.

\bibitem{jrp-soda}
Marcin Bienkowski, Jaroslaw Byrka, Marek Chrobak, Lukasz Je{\.z}, Dorian
  Nogneng, and Jir{\'{\i}} Sgall.
\newblock Better approximation bounds for the joint replenishment problem.
\newblock In {\em Proc. 25th ACM-SIAM Symp. on Discrete Algorithms (SODA)},
  pages 42--54, 2014.

\bibitem{aggregation-wads}
Marcin Bienkowski, Jaroslaw Byrka, Marek Chrobak, {\L}ukasz Je{\.z}, Ji\v{r}i
  Sgall, and Grzegorz Stachowiak.
\newblock Online control message aggregation in chain networks.
\newblock In {\em Proc. 13th Algorithms and Data Structures Symposium (WADS)},
  pages 133--145, 2013.

\bibitem{online-matching-simple}
Benjamin Birnbaum and Claire Mathieu.
\newblock On-line bipartite matching made simple.
\newblock {\em SIGACT News}, 39(1):80--87, 2008.

\bibitem{borodin-book}
Allan Borodin and Ran {El-Yaniv}.
\newblock {\em Online Computation and Competitive Analysis}.
\newblock Cambridge University Press, 1998.

\bibitem{aggregation-bkv}
Carlos Brito, Elias Koutsoupias, and Shailesh Vaya.
\newblock Competitive analysis of organization networks or multicast
  acknowledgement: {How} much to wait?
\newblock {\em Algorithmica}, 64(4):584--605, 2012.

\bibitem{multilevel-aggregation-buchbinder}
Niv Buchbinder, Moran Feldman, Joseph~(Seffi) Naor, and Ohad Talmon.
\newblock \emph{O}(depth)-competitive algorithm for online multi-level
  aggregation.
\newblock In {\em Proc. 28th ACM-SIAM Symp. on Discrete Algorithms (SODA)},
  pages 1235--1244, 2017.

\bibitem{adwords-primal-dual}
Niv Buchbinder, Kamal Jain, and Joseph Naor.
\newblock Online primal-dual algorithms for maximizing ad-auctions revenue.
\newblock In {\em Proc. 15th European Symp. on Algorithms (ESA)}, pages
  253--264, 2007.

\bibitem{jrp-online-buchbinder}
Niv Buchbinder, Tracy Kimbrel, Retsef Levi, Konstantin Makarychev, and Maxim
  Sviridenko.
\newblock Online make-to-order joint replenishment model: primal dual
  competitive algorithms.
\newblock In {\em Proc. 19th ACM-SIAM Symp. on Discrete Algorithms (SODA)},
  pages 952--961, 2008.

\bibitem{aggregation-survey}
Marek Chrobak.
\newblock Online aggregation problems.
\newblock {\em SIGACT News}, 45(1):91--102, 2014.

\bibitem{concave-matching}
Nikhil~R. Devanur and Kamal Jain.
\newblock Online matching with concave returns.
\newblock In {\em Proc. 44th ACM Symp. on Theory of Computing (STOC)}, pages
  137--144, 2012.

\bibitem{ranking-primal-dual}
Nikhil~R. Devanur, Kamal Jain, and Robert~D. Kleinberg.
\newblock Randomized primal-dual analysis of {RANKING} for online bipartite
  matching.
\newblock In {\em Proc. 24th ACM-SIAM Symp. on Discrete Algorithms (SODA)},
  pages 101--107, 2013.

\bibitem{tcp-ack-det-journal}
Daniel~R. Dooly, Sally~A. Goldman, and Stephen~D. Scott.
\newblock On-line analysis of the {TCP} acknowledgment delay problem.
\newblock {\em Journal of the ACM}, 48(2):243--273, 2001.

\bibitem{haste-makes-waste}
Yuval Emek, Shay Kutten, and Roger Wattenhofer.
\newblock Online matching: haste makes waste!
\newblock In {\em Proc. 48th ACM Symp. on Theory of Computing (STOC)}, pages
  333--344, 2016.

\bibitem{matching-delays-two-points}
Yuval Emek, Yaacov Shapiro, and Yuyi Wang.
\newblock Minimum cost perfect matching with delays for two sources.
\newblock In {\em Proc. 10th Int. Conf. on Algorithms and Complexity (CIAC)},
  2017.
\newblock To appear.

\bibitem{metrics-approximation-optimal}
Jittat Fakcharoenphol, Satish Rao, and Kunal Talwar.
\newblock A tight bound on approximating arbitrary metrics by tree metrics.
\newblock {\em Journal of Computer and System Sciences}, 69(3):485--497, 2004.

\bibitem{online-matching-line-lower-bound}
Bernhard Fuchs, Winfried Hochst{\"{a}}ttler, and Walter Kern.
\newblock Online matching on a line.
\newblock {\em Theoretical Computer Science}, 332(1--3):251--264, 2005.

\bibitem{online-matching-gupta}
Anupam Gupta and Kevin Lewi.
\newblock The online metric matching problem for doubling metrics.
\newblock In {\em Proc. 39th Int. Colloq. on Automata, Languages and
  Programming (ICALP)}, pages 424--435, 2012.

\bibitem{online-matching-pruhs}
Bala Kalyanasundaram and Kirk Pruhs.
\newblock Online weighted matching.
\newblock {\em Journal of Algorithms}, 14(3):478--488, 1993.

\bibitem{tcp-ack}
Anna~R. Karlin, Claire Kenyon, and Dana Randall.
\newblock Dynamic {TCP} acknowledgement and other stories about e/(e - 1).
\newblock {\em Algorithmica}, 36(3):209--224, 2003.

\bibitem{snoopy-caching-randomized}
Anna~R. Karlin, Mark~S. Manasse, Lyle~A. McGeoch, and Susan Owicki.
\newblock Competitive randomized algorithms for non-uniform problems.
\newblock {\em Algorithmica}, 11(6):542--571, 1994.

\bibitem{online-matching}
Richard~M. Karp, Umesh~V. Vazirani, and Vijay~V. Vazirani.
\newblock An optimal algorithm for on-line bipartite matching.
\newblock In {\em Proc. 22nd ACM Symp. on Theory of Computing (STOC)}, pages
  352--358, 1990.

\bibitem{khanna-message-aggregation}
Sanjeev Khanna, Joseph Naor, and Danny Raz.
\newblock Control message aggregation in group communication protocols.
\newblock In {\em Proc. 29th Int. Colloq. on Automata, Languages and
  Programming (ICALP)}, pages 135--146, 2002.

\bibitem{online-matching-khuller}
Samir Khuller, Stephen~G. Mitchell, and Vijay~V. Vazirani.
\newblock On-line algorithms for weighted bipartite matching and stable
  marriages.
\newblock {\em Theoretical Computer Science}, 127(2):255--267, 1994.

\bibitem{online-matching-koutsoupias}
Elias Koutsoupias and Akash Nanavati.
\newblock The online matching problem on a line.
\newblock In {\em Proc. 1st Workshop on Approximation and Online Algorithms
  (WAOA)}, pages 179--191, 2003.

\bibitem{bipartite-matching-strongly-lp}
Mohammad Mahdian and Qiqi Yan.
\newblock Online bipartite matching with random arrivals: an approach based on
  strongly factor-revealing {LP}s.
\newblock In {\em Proc. 43rd ACM Symp. on Theory of Computing (STOC)}, pages
  597--606, 2011.

\bibitem{adwords-survey}
Aranyak Mehta.
\newblock Online matching and ad allocation.
\newblock {\em Foundations and Trends in Theoretical Computer Science},
  8(4):265--368, 2013.

\bibitem{adwords-lp}
Aranyak Mehta, Amin Saberi, Umesh~V. Vazirani, and Vijay~V. Vazirani.
\newblock Adwords and generalized online matching.
\newblock {\em Journal of the ACM}, 54(5), 2007.

\bibitem{online-matching-meyerson}
Adam Meyerson, Akash Nanavati, and Laura~J. Poplawski.
\newblock Randomized online algorithms for minimum metric bipartite matching.
\newblock In {\em Proc. 7th ACM-SIAM Symp. on Discrete Algorithms (SODA)},
  pages 954--959, 2006.

\bibitem{adwords-ec}
Joseph Naor and David Wajc.
\newblock Near-optimum online ad allocation for targeted advertising.
\newblock In {\em Proc. 16th ACM Conf. on Economics and Computation (EC)},
  pages 131--148, 2015.

\bibitem{delay-sensitive-aggregation}
Yvonne~Anne Pignolet, Stefan Schmid, and Roger Wattenhofer.
\newblock Tight bounds for delay-sensitive aggregation.
\newblock {\em Discrete Mathematics {\&} Theoretical Computer Science},
  12(1):39--58, 2010.

\bibitem{matching-tarjan-reingold}
Edward~M. Reingold and Robert~Endre Tarjan.
\newblock On a greedy heuristic for complete matching.
\newblock {\em SIAM Journal on Computing}, 10(4):676--681, 1981.

\end{thebibliography}

\end{document}